\documentclass[conference]{IEEEtran}
\usepackage{latexsym}
\usepackage{amsmath}
\usepackage{amssymb}
\usepackage{amsthm}
\usepackage{epsfig}

\newtheorem{theorem}{Theorem}
\newtheorem{thm}{Theorem}

\newtheorem{lemma}[thm]{Lemma}

\newtheorem{remark}[theorem]{Remark}

\long\def\comment#1{}

\newfont{\bbb}{msbm10 scaled 800}

\newfont{\bb}{msbm10 scaled 1100}

\newcommand{\EE}{\mbox{\bb E}}

\newcommand{\av}{{\bf a}}

\newcommand{\dv}{{\bf d}}

\newcommand{\fv}{{\bf f}}
\newcommand{\gv}{{\bf g}}
\newcommand{\hv}{{\bf h}}

\newcommand{\qv}{{\bf q}}
\newcommand{\rv}{{\bf r}}

\newcommand{\tv}{{\bf t}}
\newcommand{\uv}{{\bf u}}
\newcommand{\wv}{{\bf w}}
\newcommand{\vv}{{\bf v}}
\newcommand{\xv}{{\bf x}}
\newcommand{\yv}{{\bf y}}
\newcommand{\zv}{{\bf z}}

\newcommand{\Bm}{{\bf B}}
\newcommand{\Cm}{{\bf C}}
\newcommand{\Dm}{{\bf D}}

\newcommand{\Fm}{{\bf F}}
\newcommand{\Gm}{{\bf G}}
\newcommand{\Hm}{{\bf H}}

\newcommand{\Mm}{{\bf M}}
\newcommand{\Nm}{{\bf N}}

\newcommand{\Qm}{{\bf Q}}
\newcommand{\Rm}{{\bf R}}
\newcommand{\Sm}{{\bf S}}
\newcommand{\Tm}{{\bf T}}

\newcommand{\lambdav}{\hbox{\boldmath$\lambda$}}

\newcommand{\Lambdam}{\hbox{\boldmath$\Lambda$}}

\newcommand{\defines}{{\,\,\stackrel{\scriptscriptstyle \bigtriangleup}{=}\,\,}}

\def\argmax{\operatornamewithlimits{arg\,max}}

\begin{document}

\title{Feedback Interference Alignment:\\ Exact Alignment for Three Users in Two Time Slots}

\author{\IEEEauthorblockN{Vasilis Ntranos}
\IEEEauthorblockA{USC \\
ntranos@usc.edu}
\and
\IEEEauthorblockN{Viveck R. Cadambe}
\IEEEauthorblockA{MIT / BU \\
viveck@mit.edu}
\and
\IEEEauthorblockN{Bobak Nazer}
\IEEEauthorblockA{BU\\
bobak@bu.edu}
\and
\IEEEauthorblockN{Giuseppe Caire}
\IEEEauthorblockA{USC\\
caire@usc.edu}
\IEEEoverridecommandlockouts
\IEEEcompsocitemizethanks{
\IEEEcompsocthanksitem
%This work was supported in part by ...
}}

\maketitle

\begin{abstract}
We study the three-user interference channel where each transmitter has local feedback of the signal from its targeted receiver. We show that in the important case where the channel coefficients are static, exact alignment can be achieved over two time slots using linear schemes. This is in contrast with the interference channel where no feedback is utilized, where it seems that either an infinite number of channel extensions or infinite precision is required for exact alignment. We also demonstrate, via simulations, that our scheme significantly outperforms time-sharing even at finite SNR.{ --- SUBMITTED TO ICC 2013.}
\end{abstract}

\section{Introduction}

It is widely accepted that interference is now the dominant bottleneck in wireless networks. Conventional techniques, such as time-division, perform poorly as the number of users increases. Recently, the discovery of a new technique, known as interference alignment~\cite{mmk08,cj08,jafar11}, has shown that, in certain scenarios, it is possible for each user to achieve half its interference-free rate. Since the discovery of interference alignment, an ongoing research theme in communications systems is to study the extent to which the gains predicted by theory can be realized in practice. In this paper, we contribute to this research theme through a novel utilization of feedback. In particular, two of the primary limitations of existing schemes for interference alignment are the necessity of high signal-to-noise ratios (SNRs)~\cite{cj08,mgmk09}, and the necessity of an arbitrarily large amount of time-variations and symbol extensions~\cite{cj08,ngjv12IT} to induce alignment. In this paper, we use feedback to lift these restrictions.

We will focus on the single-antenna, three-user Gaussian interference channel with flat fading (i.e., the channel remains constant over the duration of the codeword). It is unknown how to induce alignment in this scenario except in the very high SNR regime via real interference alignment~\cite{mgmk09}. As we will show, if each transmitter has \textit{local feedback} from its targeted receiver, together they can create an effective aligned channel over two adjacent time slots. This implies that each user can simultaneously achieve $1/2$ degrees-of-freedom (DoF). Furthermore, we propose an optimization framework to maximize the sum rate and show that the average achievable rate nearly reaches half the interference-free rate for the entire SNR regime. 

\subsection{Related Work}
It is well-known that feedback does not increase the capacity of point-to-point channels \cite{shannon56}. However, for multi-user networks, feedback can increase capacity. For instance, in the Gaussian multiple-access channel and broadcast channels, feedback has been shown to enlarge the capacity region \cite{ozarow84, ozarowleung84}. More recently, several groups have studied the impact of feedback on Gaussian interference channels. Suh and Tse considered the two-user Gaussian interference channel with local feedback and determined the capacity region to within two bits \cite{st11}. For interference channels with more than two users, since the capacity (or even constant gap approximations) continues to elude researchers, recent literature has focused on the DoF, both with and without feedback. Perhaps surprisingly, Cadambe and Jafar demonstrated that in the $K$-user interference channel, even global and perfect feedback cannot increase the degrees of freedom \cite{cj09ITc}. In other words, the $K$ user interference channel has a maximum of $K/2$ DoF irrespective of the presence of feedback. In light of this negative result, a natural research direction is to explore the utility of feedback to simplify signaling schemes (in particular alignment-based schemes), and to improve finite SNR performance. This research direction is especially important because existing DoF-optimal schemes need arbitrarily large bandwidth \cite{cj08}, or channel precision \cite{mgmk09}. Further, for static channels, DoF is a brittle metric since it is discontinuous at all rational channel gains \cite{eo09}. In addition, a scheme that attains the maximum DoF may still perform poorly at finite SNR. In this paper, we will use feedback to address these issues by providing a simple alignment scheme for static interference channels with significant benefits at finite SNR.
%can be viewed as the three-user extension of their setting.

The use of feedback for interference channels with \emph{time-varying} channel gains was studied in \cite{mjs12}, where it was shown that similar to the broadcast channels \cite{mt10}, feedback of stale channel information provides DoF gains. In contrast, the object of focus of our paper is the \emph{static} interference channel, where the coding is done over a single coherence block. Very recently, Papailiopoulos, Suh, and Dimakis studied the static $K$-user Gaussian interference channels under two feedback models \cite{psd12}. In the first, each transmitter gets channel output feedback from every receiver and $K/2$ DoF are easily achievable over two time-slots. In the second, the feedback is observed through a ``backwards interference channel.'' They proved that $3/2$ DoF is achievable for the $3$-user case over $4$ time-slots and presented numerical evidence for the $4$-user case. Geng and Viswanath obtained similar results for the second model and developed a generalization to full duplex communication \cite{gv12}. The key distinction between this work and our own is that we focus on local feedback, i.e., each transmitter only has channel output feedback from its intended receiver. For this setting, we will show that \emph{$3/2$ degrees of freedom are achievable with $2$ time-slots with only local feedback information.} Thus, we present a stronger version of \cite{psd12}, where global feedback was used to achieve the same degrees-of-freedom result. In addition, we optimize our scheme to significantly outperform standard techniques such as time-sharing even at finite SNRs.

%With progress in network information theory, the technique of feedback has been used in increasingly novel ways. While Sha
%nnon's classic result shows that feedback does not improve capacity on the point-to-point channel, it is well known that it can be used to improve the error exponent on this channel. In multiple access channels, feedback increases channel capac
%ity by inducing co-operation between the transmitters \cite{Ozarow}. Interestingly, in the broadcast channel, even though
%there is only a single source and there is no-possiblity of inducing co-peration, feedback can be used to increase the cap
%acity region \cite{}. This is because feedback enables transmitters to exploit side information at the receivers. In our p
%aper, we study the utilization of feedback on the $3$-user interference channel.

\section{Problem Statement}
We will focus on the $3$-user interference channel with (perfect) feedback. In this setting, there are three users indexed by $\ell \in \{1,2,3\}$. Given a \emph{blocklength} of length $n$, each transmitter has a message $w_\ell$ that is drawn independently and uniformly over $\{1,2,\ldots, 2^{nR_\ell}\}$. The channel input\footnote{For ease of analysis, we will consider real-valued channels. All of the analysis extends naturally to complex-valued channels.} from the $\ell$th transmitter is the tuple $\left(x_\ell[1],x_\ell[2],\ldots,x_\ell[n]\right) \in \mathbb{R}^{n}$ whose $i$th entry corresponds to the channel input for the $i$th symbol, where $1 \leq i \leq n$. All channel inputs must satisfy an average power constraint, $\| \frac{1}{n}\sum_{i=1}^{n}{x}_\ell(i)\|^2 \leq nP_\ell$. The channel output at the $k$th receiver is $\mathbf{y}_k= \big[ y_k[1]~y_k[2]~\cdots~ y_k[n] \big]^{\rm T}$ where 
\begin{align}
y_{k}[t] = h_{k1}x_{1}[t] +h_{k2}x_{2}[t] +h_{k3}x_{3}[t] + z_{k}[t]  
\end{align} and the noise $z_k[t]$ is i.i.d. Gaussian with mean zero and unit variance. Let $\mathbf{H} = \{h_{k\ell}\}$ denote the matrix of channel coefficients. We assume that $\mathbf{H}$ is globally available at all transmitters and receivers. (In fact, as we will argue, our scheme can tolerate a significant delay in acquiring the channel state.) The channel output can be written concisely in vector notation
$$\mathbf{y}[t] = \mathbf{H}\mathbf{x}[t] + \mathbf{z}[t] $$ where $\mathbf{y}[t]=\big[y_1[t]~y_2[t]~y_3[t]\big]^{\rm T}$, $\mathbf{x}[t] = \big[ x_1[t] ~x_2[t] ~x_3[t]\big]^{\rm T}$, and $\mathbf{z}[t] = \big[ z_1[t] ~z_2[t] ~z_3[t]\big]^{\rm T}$. In this short-hand notation, we also sometimes vectorize the power constraints as $\mathbf{P}=\big[P_1,P_2,P_3\big]^{\rm T}$.
The key assumption is that each transmitter is given access to local feedback from its intended receiver. That is, the encoding function at each transmitter can be written in terms of the message and the channel outputs up to time $t-1$, $x_\ell[t] = \mathcal{E}_{\ell,t}(w_\ell, y_\ell[1],\cdots,y_\ell[t-1])$. The decoding function at the $k$th receiver uses its $n$ channel observations to estimate the message from the $k$th transmitter, $\hat{w}_k = \mathcal{D}_k(y_k[1],\cdots,y_k[n])$.

We say a rate tuple $(R_1,R_2,R_3)$ is achievable for a given channel matrix $\mathbf{H}$ if, for any $\epsilon > 0$ and $n$ large enough, there exist encoders and decoders such that \mbox{$\mathbb{P}\big( (\hat{w}_1,\hat{w}_2,\hat{w}_3) \neq (w_1,w_2,w_3)\big) < \epsilon$}. We will be primarily concerned with determining the maximum sum rate $R_1 + R_2 + R_3$. The sum-capacity denoted by $C(\mathbf{P})$ is the supremum of $R_1+ R_2+R_3$ over the set of all achievable rate tuples $(R_1, R_2, R_3)$. When $P_1 = P_2 = P_3 = P,$ we simply denote the capacity as $C(P)$.
As usual, the sum degrees-of-freedom (DoF) are defined as
$$d \defines \lim_{P \rightarrow \infty} \frac{C(P)}{\frac{1}{2}\log{P}}.$$
\begin{remark}
Strictly speaking, the DoF $d$ is a function of the channel matrix $\mathbf{H}$. However, as we show in this paper, the DoF of the channel under consideration takes the same value\footnote{The DoF for the interference channel with feedback is not dependent on whether the channel gains are rational or irrational, unlike the case of the interference channel \emph{without} feedback~\cite{eo09}. This makes DoF a more ``robust'' performance metric for the channel in consideration.} for almost all values of $\mathbf{H}$; we therefore suppress its dependence on the channel gain matrix in this paper.
\end{remark}

%
%\section{Channel Model}
%We consider the single antenna $3$-user flat-fading interference channel given by
%
%\begin{eqnarray}
%y_{1}[n] &=& h_{11}x_{1}[n] +h_{12}x_{2}[n] +h_{13}x_{3}[n] + z_{1}[n] \nonumber \\ 
%y_{2}[n] &=& h_{21}x_{1}[n] +h_{22}x_{2}[n] +h_{23}x_{3}[n] + z_{2}[n]\nonumber \\ 
%y_{3}[n] &=& h_{31}x_{1}[n] +h_{32}x_{2}[n] +h_{33}x_{3}[n] + z_{3}[n] \\ \nonumber
%\label{eq1}
%\end{eqnarray}
%
%where 
%$x_{\ell}[n]$ denotes the encoded symbol at transmitter $\ell$ at time $n$, $h_{k\ell}$ is the channel gain coefficient from transmitter $\ell$ to receiver $k$, and $z_{k}[n]$ is the zero-mean additive white Gaussian noise observed at receiver $k$. 
%
%We assume an average power constraint on all transmitted signals ${\mathbb E}[x_{k}^{2}]\leq p$, and channel gain coefficients that remain fixed over $t$.
%
%We further assume that there are three {\it local} feedback links in the system that can be used by the receivers to feed back information to their corresponding transmitters. Each transmitter $k$ observes all the previous channel outputs of its intended receiver $Y_{i=1}^{n-1}(\ell)=\{y_{\ell}[1],\dots, y_{\ell}[n-1]\}$, and uses that information to encode $x_{\ell}[n]$ at time $n$.
%
%[{\it ...define: Message, Achievable Rate, Encoding, Decoding, Capacity, DoF}]
%
\section{Feedback Interference Alignment}
We describe here a transmission strategy that takes advantage of the locally available feedback in order to improve the overall system performance. In our scheme, transmitters send linear combinations of new symbols and available feedback. Quite remarkably, we show that the proposed scheme can align all interference over two time slots and achieve the maximum DoF of the above channel almost surely. Then we focus on the {\it finite SNR regime} and investigate efficient algorithms to optimize performance.

\subsection{The Achievable Scheme}
Our scheme employs two channel uses in order to successfully transmit one symbol to each receiver. For notational convenience, we assume these time slots are consecutive. Note that the performance of our scheme will not change if these two time slots are separated by a longer delay (up to $n/2$), owing to delay in the feedback path. Let $x_\ell[1]$ denote the first symbol from transmitter $\ell$; receiver $\ell$ intends to estimate this symbol over two time slots.
In the first slot ($t=1$), transmitter $\ell$ simply sends $\mathbf{x}_{\ell}[1].$ The $k$th receiver observes
 $y_{k}[1] = \sum_{\ell=1}^{3} h_{k\ell}x_{\ell}[1]  + z_{k}[1]$. 
Rewriting everything in the compact matrix notation introduced above, we get 
\begin{equation}
\yv[1]=\Hm \xv[1] + \zv[1]
\label{slot1}
\end{equation}
%where $\xv = [x_{1}, x_{2}, x_{3}]^{\rm T}$, $\yv[t] = [ y_{1}[t], y_{2}[t], y_{3}[t]]^{\rm T}$,
%$\zv[t] = [ z_{1}[t], z_{2}[t], z_{3}[t]]^{\rm T}$ and 
These noisy observations are then fed back by the receiver to the corresponding transmitters.
In the second time slot ($t=2$), transmitter $\ell$ observes $y_{\ell}[1]$ and transmits the linear combination $x_{\ell}[2] = t_{\ell}x_{\ell}[1] + f_{\ell}y_{\ell}[1]$  
for some $t_{\ell},f_{\ell} \in \mathbb R$ whose choice will soon be described. Let $\tv \triangleq [t_{1}, t_{2}, t_{3}]^{\rm T}$, $\fv \triangleq [f_{1}, f_{2}, f_{3}]^{\rm T}$
and $\Tm={\rm diag}(\tv)$, $\Fm={\rm diag}(\fv)$\footnote{As in MATLAB, ${\rm diag}(\dv)$ is a diagonal matrix formed by the vector $\dv$. $\{{\rm diag}(\dv)\}_{i,i}=d_{i}$ and $\{{\rm diag}(\dv)\}_{i,j\neq i}= 0$.}.
In matrix notation, the receivers observe 
\begin{eqnarray}
\yv[2] &=& \Hm\left(\Tm \xv[1] +\Fm\yv[1]\right) +\zv[2] \nonumber \\
&=& \Hm\left(\Tm \xv[1] +\Fm(\Hm \xv[1] + \zv[1])\right) +\zv[2] \nonumber \\
&=& \Hm\left(\Tm+\Fm\Hm\right) \xv[1] + \Hm\Fm\zv[1] +\zv[2]  \nonumber
\label{slot2}
\end{eqnarray}

\begin{remark} Note that our scheme only makes use of the channel state information $\mathbf{H}$ during the second phase. Therefore, we can tolerate a delay of up to $n/2$ in learning the channel matrix.
\end{remark}

Now, after the second transmission, each receiver collects its received signals to form a two-dimensional observation 
$\rv_{k} \triangleq [y_{k}[1], y_{k}[2]]^{\rm T}$. Let $\hv_{k}^{\rm T}$ denote the 
$k$th row of the channel matrix $\Hm$. The equivalent channel observed at receiver $k$ can be written as
\begin{eqnarray}
\rv_{k}  &=& 
\left[
    \begin{array}{c}
      \hv_{k}^{\rm T}\xv[1] + z_{k}[1] \\
      \hv_{k}^{\rm T}(\Tm +\Fm\Hm)\xv[1] + \hv_{k}^{\rm T}\Fm\zv[1] +z_{k}[2]
    \end{array}
  \right]
  \\&=& 
\left[
    \begin{array}{c}
      \!\!\hv_{k}^{\rm T} \\
      \!\!\hv_{k}^{\rm T}(\Tm +\Fm\Hm)\!\!
    \end{array}
  \right]\xv[1]+
  \left[
    \begin{array}{c}
      z_{k}[1] \\
     \!\! \hv_{k}^{\rm T}\Fm\zv[1] +z_{k}[2]\!\!
    \end{array}
  \right]
  \\&=&
  \Gm_{k} \xv[1] + \wv_{k}
\end{eqnarray}
where $\Gm_{k} = \left[\begin{array}{c}\mathbf{h}_{k}^{\rm T} \\ \mathbf{h}_{k}^{\rm T}(\mathbf{T}+\mathbf{F}\mathbf{H}) \end{array}\right]\in {\mathbb R}^{2\times3}$ is the equivalent channel matrix seen by receiver $k$  and $\wv_{k}$ is the zero-mean colored Gaussian noise vector whose covariance matrix is given by
\begin{eqnarray}
\Cm_{k}&=&{\mathbb E}[\wv_{k}\wv_{k}^{\rm T}] 
=\begin{bmatrix}
       1 & h_{kk}f_{k}    \\[0.3em]
       h_{kk}f_{k} & ||\hv_{k}^{\rm T}\Fm ||^{2}+1        
     \end{bmatrix} 
     \label{cov}
\end{eqnarray}

Notice that for all $k$, both the covariance matrix $\Cm_{k}$ and the equivalent channel matrix $\Gm_{k}$ depend on the  choice of the linear combination coefficients $\tv$,$\fv\in {\mathbb R }^{3}$.

Now, for a given choice of $\tv$ and $\fv$, the $k$th receiver can use a linear MMSE filter to estimate $x_{k}$ from the  vector observation $\rv_{k}$. 
Assuming real-valued channels, Gaussian signaling and coding over a long block, the rate achieved (after normalizing for the use of two time slots) for the $k$th user is 
\begin{equation}
R_{k}(\tv,\fv) = \frac{1}{4}\log\left(\frac{ \mbox{det} \left( \Cm_{k} + \sum_{\ell}P_\ell\gv_{\ell,k}\gv_{\ell,k}^{\rm T}\right)}{\mbox{det}\left({\Cm_{k} + \sum_{\ell\neq k}P_\ell\gv_{\ell,k}\gv_{\ell,k}^{\rm T}}  \right)}\right)
\label{Rmmse}
\end{equation}where $\gv_{\ell,k}$ is the $\ell$th column of the equivalent channel matrix $\Gm_{k}$.

We now intend to describe the set of all feasible $(t_\ell, f_\ell),\; \ell=1,2,3$ that satisfy the power constraints. While the analysis that follows can be performed for arbitrary power constraints, for simplicity of notation, we describe this region for the case of $P_1=P_2=P_3=P.$ Let ${\cal P_{\ell}}$ be the set of all linear combination coefficients $t_{\ell},f_{\ell}$ satisfying the average power constraint 
$\displaystyle\mathbb E \left[(x_{\ell}[2])^{2}\right]\leq P$. We have
\begin{align}
&{\mathbb E}\left[(x_{\ell}[2])^{2}\right] \\
&= {\mathbb E}\left[(t_{\ell}x_{\ell}[1] + f_{\ell}y_{\ell}[1] )^{2}\right] \nonumber \\
&=t_{\ell}^{2}{\mathbb E}[(x_{\ell}[1])^{2}] + f_{\ell}^{2}{\mathbb E}\left[(y_{\ell}[1])^{2}\right] + 2t_{\ell}f_{\ell}{\mathbb E}\left[x_{\ell}[1]y_{\ell}[1]\right] \nonumber \\
&=t_{\ell}^{2}P + f_{\ell}^{2}\left(P||\hv_{\ell}^{\rm T}||^{2}+1\right) + 2t_{\ell}f_{\ell}h_{\ell\ell}   
\end{align}
and
\begin{eqnarray}\small
{\cal P_{\ell}} = \!\left\{\begin{bmatrix}
       t_{\ell}  \\[0.3em]
       f_{\ell}         
       
     \end{bmatrix} \! \! \in \mathbb R^{2} \!  :\!  \left|\left| \begin{bmatrix}
       1 & h_{\ell\ell}    \\
       0 & \sqrt{{P^{-1} + \sum_{j\neq \ell} h_{\ell j}^{2}}}        
       
     \end{bmatrix}
     \begin{bmatrix}
       t_{\ell}  \\[0.3em]
       f_{\ell}         
       
     \end{bmatrix}\right|\right|\leq 1 \!
 \right\}
 \label{power}
\end{eqnarray}

Let $\displaystyle {\cal P} = \left\{ (\tv,\fv)\in {\mathbb R}^{3}\times {\mathbb R}^{3}: (t_{\ell},f_{\ell})\in{\cal P}_{\ell}, \forall \ell \right\}$. The achievable sum rate of the proposed scheme is 

\begin{equation}
R_{sum} = \max_{(\tv,\fv)\in {\cal P}}\sum_{k}R_{k}(\tv,\fv)
\label{sumrate}
\end{equation} 

\subsection{Interference Alignment and DoF}
The natural strategy for the  achievable scheme described above would be to solve (\ref{sumrate}). However, the associated optimization problem is non-convex and hard to analyze. In this section, to simplify analysis, we use the coarser metric of DoF in order to study the performance of our scheme. 

In particular, we will show that it is possible, through an appropriate choice of $\tv$ and $\fv$, for receiver $k$ to null the interference and recover $\lambda_{k}x_{k}[1] + \tilde{z}_l$ where $\lambda_{k} \neq 0$ and $\tilde{z}_{k}$ depends only on $\wv_{k}$. Since we use two time-slots to achieve the above, this would automatically imply that receiver $k$ can decode $x_k[1]$ at $1/2$ DoF as required. The interference alignment is implicit in this scheme because receiver $k$ receives (noisy) linear combinations of $3$ scalars, $x_{1}[1], x_{2}[1],x_{3}[1]$ in two time-slots; resolving $x_{k}[1]$ is possible if the two interferers $x_{j}[1],\, j \neq \ell$ align. To obtain $x_{k}[1],$ receiver $k$ linearly combines $y_{k}[1],y_{k}[2]$ as $y_k[2]+d_{k} y_k[1].$ The goal is to choose $\dv \triangleq [d_{1}, d_{2}, d_{3}]^{\rm T}$, $\tv$, $\fv$ so that
$$y_k[2] + d_k y_k[1] =  \lambda_k x_k[1] + \tilde{z}_{k},\;\; \lambda_{k}\neq 0.$$ 
The above condition for decodability can be expressed in matrix form as 
\begin{equation}
\Dm\Hm+\Hm\Tm+\Hm\Fm\Hm = \Lambdam, ~~{\rm det}(\Lambdam)\neq 0
\label{eq:alignment}
\end{equation} 
where $\mathbf{D}={\rm diag}(\dv).$ 
\begin{remark}
One can think of the vector $[1, d_{k}]$ as a (scaled) linear projection of the corresponding  2-dimensional observation at receiver $k$.
It is worth noting that for fixed $\tv$, $\fv$, the optimal linear receiver that combines $y_{k}[1],y_{k}[2]$ to recover $x_{k}[1]$, assuming that the interference is treated as noise, is given by the MMSE receiver. The corresponding rate achieved is indicated in (\ref{Rmmse}). However, for the purposes of analyzing the coarser metric of degrees of freedom, it suffices to choose the linear combination coefficients $d_{k}$ to zero-force the interference, i.e., to satisfy 
(\ref{eq:alignment}).
\end{remark}
We next show that we can indeed design $\Dm, \Tm, \Fm$ and $\Lambdam$ so that the above condition is satisfied. For space considerations, we provide here only a sketch of the proof.
\begin{lemma}
There exist diagonal matrices $\Dm$, $\Tm$, $\Fm$ and $\Lambdam$ $\in \mathbb R^{3\times3}$ such that
$\Dm\Hm+\Hm\Tm+\Hm\Fm\Hm = \Lambdam$, ${\rm det}(\Lambdam)\neq 0$ with probability 1.
\label{lemma1}
\end{lemma}
%\begin{proof}
%See the long version of this paper [arxiv].
%\end{proof}
\begin{proof}
Let 
$d_{i}$, $t_{i}$, $f_{i}$, $\lambda_{i}\in \mathbb R$ be the diagonal elements of $\Dm$, $\Tm$, $\Fm$ and $\Lambdam$ respectively and let
\begin{eqnarray}
\av &\triangleq& [d_{1}, d_{2}, d_{3}, t_{1}, t_{2}, t_{3}, f_{1}, f_{2}, f_{3} ]^{\rm T}, \nonumber \\ 
\lambdav &\triangleq& [\lambda_{1}, 0, 0, 0, \lambda_{2}, 0, 0, 0, \lambda_{3} ]^{\rm T}.
\end{eqnarray}
%
%Let $\mathbf{a} \triangleq [d_{1}, d_{2}, d_{3}, t_{1}, t_{2}, t_{3}, f_{1}, f_{2}, f_{3} ]^{\rm T}$ and 
%$\lambdav \triangleq [\lambda_{1}, 0, 0, 0, \lambda_{2}, 0, 0, 0, \lambda_{3} ]^{\rm T}$,
%where $d_{i}$, $t_{i}$, $f_{i}$, $\lambda_{i}\in \mathbb R$ are the diagonal elements of $\Dm$, $\Tm$, $\Fm$ and $\Lambdam$ respectively.

Note that, for any fixed value of $\mathbf{\Lambda},$ the interference alignment condition $\Dm\Hm+\Hm\Tm+\Hm\Fm\Hm = \Lambdam$ is essentially a set of linear equations in $\mathbf{a}$ with one equation corresponding to each entry of $\Lambdam$. Also, note that the condition ${\rm det}(\Lambdam)\neq 0$ is equivalent to stating that $\lambda_{1}\lambda_{2}\lambda_{3} \neq 0.$ Therefore, the conditions of the lemma can be equivalently written in terms of $\mathbf{a}$ and $\lambdav$ as

\begin{equation}
\left\{\begin{array}{l l}
    \Mm\mathbf{a}=\lambdav\\
    \lambda_{1}\lambda_{2}\lambda_{3}\neq 0\\
  \end{array}\right.
  \label{cond1}
\end{equation}
%\Mm\xv=\lambdav 
where 
\begin{align}&\Mm = \nonumber \\
&\small \begin{bmatrix}
h_{11} & 0 & 0 & h_{11} & 0 & 0 & h_{11}^{2} & h_{12}h_{21} & h_{13}h_{31}\\
0 & h_{21} & 0 & h_{21} & 0 & 0 & h_{21}h_{11} & h_{22}h_{21} & h_{23}h_{31}\\
0 & 0 & h_{31} & h_{31} & 0 & 0 & h_{31}h_{11} & h_{32}h_{21} & h_{33}h_{31}\\
h_{12} & 0 & 0 & 0 & h_{12} & 0 & h_{11}h_{12} & h_{12}h_{22} & h_{13}h_{32}\\
0 & h_{22} & 0 & 0 & h_{22} & 0 & h_{21}h_{12} & h_{22}^{2} & h_{23}h_{32}\\
0 & 0 & h_{32} & 0 & h_{32} & 0 & h_{31}h_{12} & h_{32}h_{22} & h_{33}h_{32}\\
h_{13} & 0 & 0 & 0 & 0 & h_{13} & h_{11}h_{13} & h_{12}h_{23} & h_{13}h_{33}\\
0 & h_{23} & 0 & 0 & 0 & h_{23} & h_{21}h_{13} & h_{22}h_{23} & h_{23}h_{33}\\
0 & 0 & h_{33} & 0 & 0 & h_{33} & h_{31}h_{13} & h_{32}h_{23} & h_{33}^{2}
\end{bmatrix} \nonumber
\end{align}

For convenience, we can divide the $9$ linear equations above into two parts: the first part contains $6$ equations whose right-hand side should evaluate to $0,$ and the second part contains the remaining $3$ equations whose right-hand side is equal to a non-zero value $\lambda_{i}$. Denote as $\Bm\in \mathbb R^{6\times 9}$ the rows of $\Mm$ that correspond to the first part, and let $\Qm\in \mathbb R^{3\times 9}$ be the matrix containing the remaining rows that correspond to  $\lambda_{1}$,$\lambda_{2}$,$\lambda_{3}$. Condition (\ref{cond1}) becomes

\begin{equation}
\left\{\begin{array}{l l}
    \Bm\mathbf{a}=0\\
    \Qm\mathbf{a}=\begin{bmatrix}\lambda_{1} & \lambda_{2} & \lambda_{3}\end{bmatrix}^{\rm T}\\
    \lambda_{1}\lambda_{2}\lambda_{3}\neq 0\\
  \end{array}\right.
  \label{cond2}
\end{equation}

Let ${\cal N}(\Bm)$ denote the nullspace of $\Bm$. By performing elementary row operations, it can be verified 
%(see \cite{Ntranos_etal_journal} for more details) 
that the following matrix is a basis for the nullspace of $\mathbf{B}$.
%  t%We are going to show that there exists a vector $\vv\in{\cal N}(\Bm)$ such that
%$\Qm \vv$ produces a vector $\cv\in {\mathbb R^{3}}$ with non-zero coordinates.
%Assuming that $h_{ij}\neq 0, \forall i,j$ and 
%$h_{13}(h_{21}h_{32}+h_{23}h_{32})\neq h_{23}(h_{11}h_{32}+h_{12}h_{31})$,
%the reduced row echelon form of $\Bm$ is given by 
%$$\Bm_{r} = 
%\begin{bmatrix}
%1 & 0 & 0 & 0 & 0 & 1  & 0 & \frac{h_{23}(h_{12}h_{31}-h_{11}h_{32})}{h_{13}h_{31}} & h_{33}-\frac{h_{11}h_{23}h_{32}}{h_{12}h_{21}}\\[0.4em]
%0 & 1 & 0 & 0 & 0 & 1  & 0 & h_{22}-\frac{h_{21}h_{32}}{h_{31}}                      & h_{33}-\frac{h_{13}h_{32}}{h_{12}}\\[0.4em]
%0 & 0 & 1 & 0 & 0 & 1  & 0 & 0 													 & \frac{2h_{12}h_{33}-h_{13}h_{32}}{h_{12}}-\frac{h_{23}h_{31}}{h_{21}}\\[0.4em]
%0 & 0 & 0 & 1 & 0 & -1 & 0 & \frac{h_{32}(h_{13}h_{21}-h_{11}h_{23})}{h_{13}h_{31}} & \frac{h_{13}h_{32}-h_{12}h_{33}}{h_{12}}-\frac{h_{23}(h_{11}h_{32}-h_{12}h_{31})}{h_{12}h_{21}}\\[0.4em]
%0 & 0 & 0 & 0 & 1 & -1 & 0 & h_{22}-\frac{h_{12}h_{23}}{h_{13}} 					 & \frac{h_{13}h_{32}}{h_{12}}-h_{33}\\[0.4em]
%0 & 0 & 0 & 0 & 0 & 0  & 1 & \frac{h_{23}h_{32}}{h_{13}h_{31}} 						 & \frac{h_{23}h_{32}}{h_{12}h_{21}} 
%\end{bmatrix}
%$$ 
\begin{align}
&\Nm=\nonumber \\
&\small \begin{bmatrix}
-1  &\! \frac{h_{23}(h_{11}h_{32}-h_{12}h_{31})}{h_{13}h_{31}} & \frac{h_{11}h_{23}h_{32}}{h_{12}h_{21}}-h_{33}\\[0.3em]
-1  &\! \frac{h_{21}h_{32}}{h_{31}} -h_{22}& \frac{h_{13}h_{32}}{h_{12}}-h_{33}\\[0.3em]
-1  &\! 0 &  \frac{h_{23}h_{31}}{h_{21}}-\frac{2h_{12}h_{33}-h_{13}h_{32}}{h_{12}}\\[0.3em]
1 &\! \frac{h_{32}(h_{11}h_{23}-h_{13}h_{21})}{h_{13}h_{31}} 
& \frac{h_{12}h_{33}-h_{13}h_{32}}{h_{12}}+\frac{h_{23}(h_{11}h_{32}-h_{12}h_{31})}{h_{12}h_{21}}\\[0.3em]
1 &\! \frac{h_{12}h_{23}}{h_{13}}-h_{22} & h_{33}-\frac{h_{13}h_{32}}{h_{12}}\\[0.3em]
1 &\! 0 & 0\\[0.3em]
0  &\! -\frac{h_{23}h_{32}}{h_{13}h_{31}}  & -\frac{h_{23}h_{32}}{h_{12}h_{21}}\\[0.3em]
0  &\! 1 & 0\\[0.3em]
0  &\! 0 & 1\\
\end{bmatrix} \nonumber
\end{align}
Any $\vv\in{\cal N}(\Bm)$ can be written as $\Nm\wv$, $\wv \in \mathbb R^{3}$. Condition (\ref{cond2}) becomes 

\begin{equation}
\left\{\begin{array}{l l}
	\wv \in \mathbb R^{3}\\
    \Qm\Nm\wv=\begin{bmatrix}\lambda_{1} & \lambda_{2} & \lambda_{3}\end{bmatrix}^{\rm T}\\
    \lambda_{1}\lambda_{2}\lambda_{3}\neq 0\\
  \end{array}\right.
  \label{cond3}
\end{equation}
We have 
\begin{align} 
\hspace{-0.13in}
%&\mathbf{S} \defines \mathbf{Q}\mathbf{N}= \nonumber\\
%&\small
\mathbf{Q}\mathbf{N}= \begin{bmatrix}
0 & s_{12} &s_{13}\\
0 & s_{22} &0 \\
0 & 0 & s_{33}
\end{bmatrix}
\end{align}where 
\begin{align}
s_{12} &= \frac{(h_{11}h_{23}-h_{13}h_{21})(h_{11}h_{32}-h_{12}h_{31})}{h_{13}h_{31}} \\
s_{13} &= \frac{(h_{11}h_{23}-h_{13}h_{21})(h_{11}h_{32}-h_{12}h_{31})}{h_{12}h_{21}} \\
s_{22} &= \frac{(h_{12}h_{23}-h_{13}h_{22})(h_{21}h_{32}-h_{22}h_{31})}{h_{13}h_{31}} \\
s_{33} &=\frac{(h_{13}h_{32}-h_{12}h_{33})(h_{21}h_{33}-h_{23}h_{31})}{h_{12}h_{21}}
\end{align} If $s_{22},s_{33}\neq 0$, $s_{12}s_{33}+s_{13}s_{22}\neq 0$, one can choose 
\mbox{$\wv^{*}=[1, {1}/{s_{22}}, {1}/{s_{33}}]^{\rm T}$} in order to obtain

%$$\Sm\wv^{*} = \begin{bmatrix}
%\frac{s_{12}}{s_{22}} + \frac{s_{13}}{s_{33}}\\
%1\\
%1\\
%\end{bmatrix}$$

$$\displaystyle \Qm\Nm\wv^{*} = \left[\frac{s_{12}}{s_{22}} + \frac{s_{13}}{s_{33}},\;1,\;1\right]^{\rm T} \;\;\;\;\mbox{and}$$
 $$\lambda_{1}\lambda_{2}\lambda_{3} = \frac{s_{12}}{s_{22}} + \frac{s_{13}}{s_{33}} = \frac{s_{12}s_{33}+s_{13}s_{22}}{s_{22}s_{33}}\neq 0.$$
Putting everything together, $\av^{*}=\Nm\wv^{*}$ satisfies  (\ref{cond1}) if
\begin{equation}
\left\{\begin{array}{l l}
	h_{ij}\neq 0, \forall i,j \\ 
h_{13}(h_{21}h_{32}+h_{23}h_{32})\neq h_{23}(h_{11}h_{32}+h_{12}h_{31})\\
	s_{22},s_{33}\neq 0, \\s_{12}s_{33}+s_{13}s_{22}\neq 0\\
  \end{array}\right.
  \label{cond4}
\end{equation}
We have a solution when the conditions (\ref{cond4}) are satisfied. Since these conditions are satisfied almost surely when $h_{ij}$ are drawn from a continuous non-degenerate distribution, the lemma is proved.

\end{proof}

\subsection{Optimization at finite SNR}

Interference alignment requires in (\ref{cond2}) that $\Bm\av=\mathbf{0}$ and Lemma \ref{lemma1} guarantees  that $\lambda_{1}$,$\lambda_{2}$,$\lambda_{3}\neq 0$ for some $\av_{0}\in {\cal N}(\Bm)$ with probability 1.
However, at finite SNR, satisfying condition (\ref{cond2}) is not enough: It is more important to be able to control the magnitude of the effective channel gains $|\lambda_{i}|^{2}$ and  optimize performance under the given transmit power constraint.

\subsubsection{Exact Interference Alignment }
Maximizing the sum-rate under the exact interference alignment condition (\ref{cond2}) corresponds to solving the following optimization problem

\begin{equation*}
\begin{aligned}
Q1: \;\;\;\;\;\;& \underset{\displaystyle \av}{\text{maximize}} 
&  & \sum_{k}\log(1+P{|\qv_{k}^{\rm T}\av|^{2}}/{\sigma^{2}_{k}(\av)}) \\
& \text{subject to:}
& & \;\;\;\;\Bm\av=0\\
& & & \;\;\;\;\; \av \in {\cal P}
\end{aligned}
\end{equation*}
where $\qv_{k}$ denotes the $k$th row of $\Qm$, ${\cal P}=\{\av \in \mathbb R^{9} : (a_{\ell+3},a_{\ell+6})\in {\cal P}_{\ell},\; \ell = 1,2,3 \}$ is the corresponding power 
constraint defined in (\ref{power}) and $\sigma^{2}_{k}(\av)$ is the effective noise variance
observed at the $k$th user that depends on $\av$ via (\ref{cov}).

As in most rate optimization problems, the main difficulty in solving $Q1$ efficiently comes directly from the non-convexity of the objective. 
We relax this problem to maximizing $\sum |\lambda_{i}|^{2}= ||\Qm\av||^{2}$ subject to the same constraints. 

Let $\Nm\in \mathbb R^{9\times3}$ denote a basis for the nullspace of $\Bm$. A simple change of variables $\av= \Nm \wv$, $\wv \in \mathbb R^{3}$ results in the following optimization problem

\begin{equation*}
\begin{aligned}
Q2: \;\;\;\;\;\;& \underset{\displaystyle \wv \in {\cal P}'}{\text{maximize}} 
&  & ||\Qm\Nm\wv||^{2} \\
\end{aligned}
\end{equation*}
where ${\cal P}' = \{\wv \in \mathbb R^{3} : \Nm\wv \in {\cal P}\}$. Tracing back  the definitions of ${\cal P}$ and ${\cal P}_{\ell}$, one can see that 
$\cal P'$ can be written as the intersection of three ellipsoids -- one for each transmit power constraint in
(\ref{power}).  
Even though maximizing a quadratic over the intersection of $k$ ellipsoids is still an intractable problem, it admits a natural semidefinite programming (SDP) relaxation with provable approximation guarantees \cite{nrt99}. 

A much simpler alternative comes directly from the singular value decomposition (SVD) of $\Sm\triangleq\Qm\Nm$. One can first choose the direction of $\wv_{0}$ along the principal component of $\Sm$, maximizing $||\Sm\wv||^{2}$ subject to $||\wv||=1$, and then scale by $c\in \mathbb R$ such that $c\cdot\wv_{0}\in {\cal P}'$. 

\subsubsection{A ``max-SINR'' approach}
 
Even though interference alignment attains the optimal DoF, it can be too restrictive in the practical SNR range. 
Here, we remove the exact alignment requirement and propose a ``departing direction'' from the nullspace of $\Bm$ (perfect IA) to its orthogonal complement in order to search for better solutions based on a max-SINR heuristic.

Let $\Rm\in \mathbb R^{9\times6}$ be a basis for the rowspace of $\Bm$, denoted ${\cal R}(\Bm^{\rm T})$, $\av_{N}\in{\cal N}(\Bm)$ be an interference alignment solution and $\vv \in {\cal R}(\Bm^{\rm T})$ be a unit vector in an orthogonal direction. We would like to explore solutions of the form $\av = \av_{N} + \delta \cdot \vv$, for some $\delta \in \mathbb R$. Any direction $\vv$ we choose will definitely increase the aggregate interference power since we are moving away from alignment. In this context, we would like to choose $\vv$ in a ``max-SINR'' direction, that can favor useful signal  more than interference. 

A natural choice for a heuristic would be to seek for $\vv\in {\cal R}(\Bm^{\rm T})$ that maximizes the ratio $||\Qm\vv||/||\Bm\vv||$. Writing $\vv = \Rm\wv$, $\wv \in \mathbb R^{6}$ we want to find

\begin{eqnarray}
\wv^{*} &=& \argmax_{\wv} \frac{||\Qm\Rm\wv||}{||\Bm\Rm\wv||} \nonumber\\
&=&\argmax_{\wv}
\frac{\wv^{\rm T}\Rm^{\rm T}\Qm^{\rm T}\Qm\Rm\wv}{\wv^{\rm T}\Rm^{\rm T}\Bm^{\rm T}\Bm\Rm\wv}.\\ \nonumber
\end{eqnarray} This can be seen as a generalized eigenvector problem in which the optimal $\wv^{*}$ is given by the principal eigenvector of the matrix $(\Rm^{\rm T}\Bm^{\rm T}\Bm\Rm)^{-1}\Rm^{\rm T}\Qm^{\rm T}\Qm\Rm$ and 
$\vv^{*}=\Rm\wv^{*}$.

Now, for any solution $\av_{N}\in {\cal N}(\Bm)$ we can use $\vv^{*}$ to search for rate maximizing solutions in their two dimensional span. In fact, one only needs to search for the ``right'' angle between $\av_{N}$ and $\vv^{*}$. Let $\uv^{*}\triangleq \av_{N}/||\av_{N}||$. The candidate solutions to the rate optimization problem can be parametrized by $\theta \in [0,\pi)$ as 

\begin{equation}
\av(\theta) = \beta(\theta)\left[\uv^{*}\cos(\theta)+\vv^{*}\sin(\theta)\right]
\label{theta}
\end{equation}
where $\beta(\theta)\triangleq \sup \{ c\in \mathbb R^{+} : c [\uv^{*}\cos(\theta)+\vv^{*}\sin(\theta)] \in {\cal P}\}$ is the appropriate scaling parameter to satisfy the power constraints and can be computed in closed form given the tractable structure of $\cal P$. 
In view of (\ref{theta}), we can then write the achievable rate in  (\ref{sumrate}) as a 
function of a single variable $\theta \in [0,\pi)$, denoted as $R_{sum}(\av(\theta))$, and find

\begin{equation}
\theta^{*} = \argmax_{\theta}R_{sum}(\av(\theta)).   
\end{equation}
Notice that the above parametrization does not exclude perfect IA solutions ($\theta = 0$) and one expects 
that $\theta^{*}$ will go to zero as the SNR increases.

We propose next an efficient heuristic algorithm that can be used to optimize the sum rate of our feedback scheme across the entire SNR range. The algorithm takes as input the matrices $\Bm\in\mathbb R^{6\times9}$ and $\Qm\in\mathbb R^{3\times9}$  (defined in the previous section) and computes the  linear combination coefficients $\tv$, $\fv\in\mathbb R^{3}$ that will be used by the transmitters in the second time-slot.

\vspace{0.1in}
\hspace{-0.11in }{{\bf{Algorithm}}: Max-SINR Feedback}\\
\vspace{-0.1in}

 {\bf Input:} matrices $\Bm\in\mathbb R^{6\times9}$, $\Qm \in\mathbb R^{3\times9}$
 
  \underline{{\bf Output:} linear combination coefficients $\tv,\fv \in \mathbb R^{3}$}\\
 
1)  {$\Nm =$ basis for $\cal N(\Bm)$, $\Rm =$ basis for $\cal R(\Bm^{\rm T})$};
 
2) $\wv_{N} =$ max\_eigenvector of $(\Qm\Nm)^{\rm T}\Qm\Nm$;
 
3) $\wv_{R} =$ max\_eigenvector of $(\Rm^{\rm T}\Bm^{\rm T}\Bm\Rm)^{-1}(\Qm\Rm)^{\rm T}\Qm\Rm$;
 
4)  $\uv^{*} = \Nm\wv_{N}/||\Nm\wv_{N}||$, $\vv^{*} = \Rm\wv_{R}/||\Rm\wv_{R}||$;
 
5) $\theta^{*} ={\displaystyle \argmax_{\theta}}R_{sum}(\beta(\theta)\left[\uv^{*}\cos(\theta)+\vv^{*}\sin(\theta)\right])$;

6) $\av = \beta(\theta^{*})\left[\uv^{*}\cos(\theta^{*})+\vv^{*}\sin(\theta^{*})\right]$;

7) $\tv=[a_{4},a_{5},a_{6}]^{\rm T}$, $\fv=[a_{7},a_{8},a_{9}]^{\rm T}$;

\section{Numerical Results}
In this section, we evaluate the performance of the proposed feedback scheme in terms of its  average achievable sum rate. We consider real channel coefficients $\Hm$
(that remain fixed for at least two time-slots) and compute the {\it average} sum rate over the ensemble of channel realizations with $h_{k\ell}\sim {\cal N}(0,1)$.
We compare  with the following four schemes.

\subsubsection*{Time Sharing}
Transmitters use the channel in a round robin fashion.
The average achievable sum rate in this case is given by 
$R_{sum} = \sum_k \frac{1}{6}\EE[\log(1+3P|h_{kk}|^{2})]$.

\subsubsection*{Treat as Noise}
All the transmitters use the channel simultaneously and interference is treated as noise.
The average sum rate is given by 
$R_{sum} =\sum_{k}  \frac{1}{2}\EE\Big[\log\Big(1+\frac{P|h_{kk}|^{2}}{1+\sum_{\ell\neq k}P|h_{k\ell}|^{2}}\Big)\Big]$.

\subsubsection*{Ergodic Alignment} For time-varying channels, this scheme can achieve half the interference-free rate at any SNR by carefully pairing up channel matrices to cancel out interference~\cite{ngjv12IT}. We use the performance of this scheme as a benchmark, even though it is \textit{not achievable} in our problem setting as the channel is static. The average sum rate is given by $R_{sum} = \sum_k \frac{1}{4} \EE[\log(1+2P|h_{kk}|^{2})]$. 
%Notice that achieving the above rate comes at the cost of an enormous decoding delay.

\subsubsection*{2-user Feedback} We consider a time sharing version of the two bit gap scheme 
proposed in \cite{st11} for the two-user interference channel with feedback. Transmitters 
can use the channel in pairs and transmit with  power $3P/2$.

In our simulations, we plot the average achievable sum rates for all of the above schemes versus  SNR measured in dB. Figure~\ref{fig1} shows the performance comparison in the standard setting where
$h_{k\ell}\sim {\cal N}(0,1)$.
We observe that the proposed max-SINR feedback scheme gets -- in two time slots -- rates that are not too far from what ergodic IA could ultimately achieve across many independent channel realizations (if the channel were time-varying).
Notice that the max-SINR heuristic is able to provide considerable gains in the low to medium SNR range while it maintains the the right slope as SNR increases.

Figure~\ref{fig2} shows the same comparison in a weaker interference power regime, where the cross channel gains are scaled such that 
$10\log_{10}\EE [|h_{k\ell}|^{2}]=-3$dB.
In this setting, the performance gains obtained by feedback become more significant, especially for the max-SINR design that seems to achieve the right balance between useful signal and interference power.

\begin{figure}[h]
\centerline{\includegraphics[width=1.05\columnwidth]{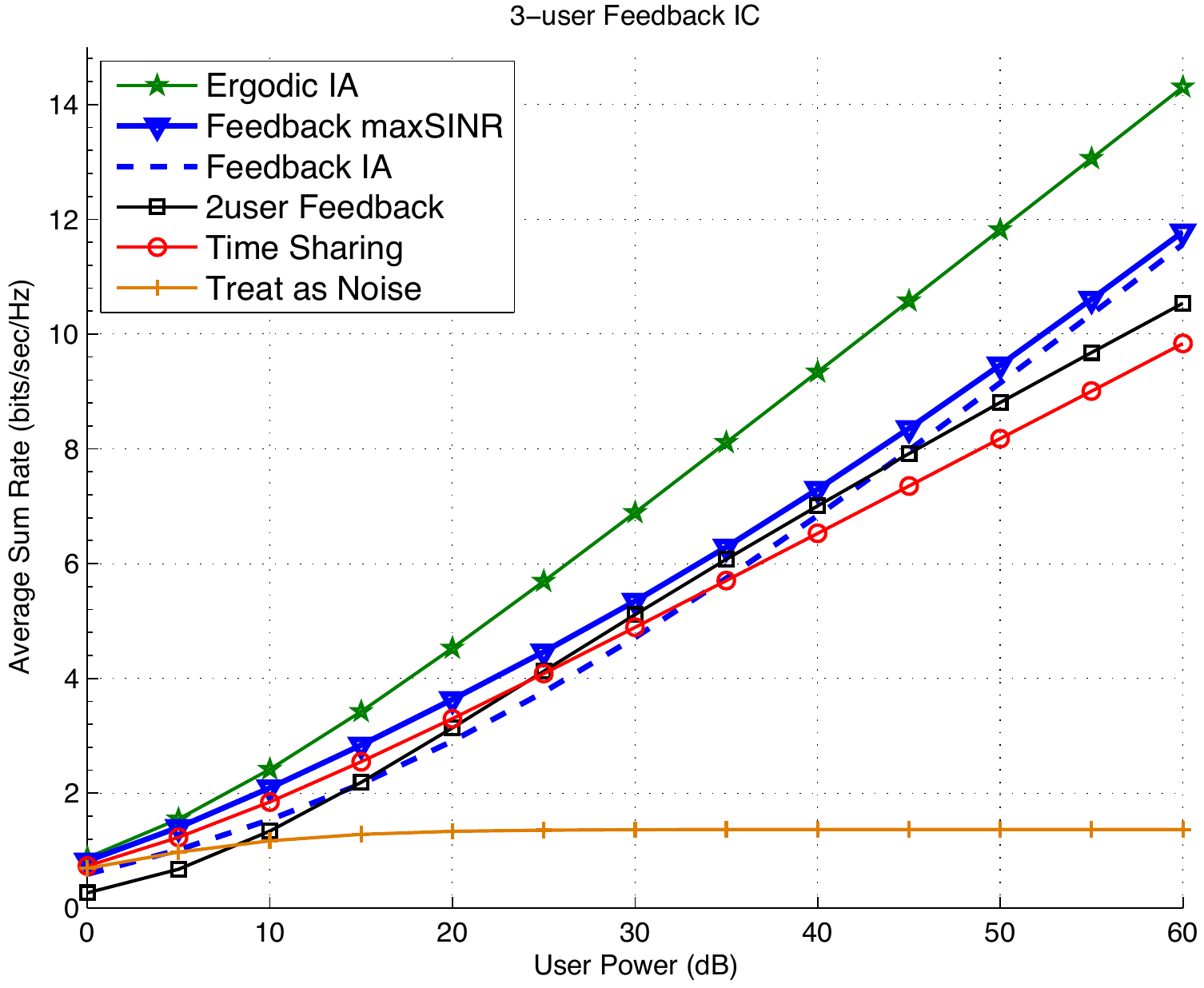}}
\caption{Performance of the proposed feedback scheme. We assume real channel coefficients $h_{k\ell}\sim {\cal N}(0,1)$ and we plot average sum rates.  }
\label{fig1}
\end{figure}

\begin{figure}[h]
\centerline{\includegraphics[width=1.05\columnwidth]{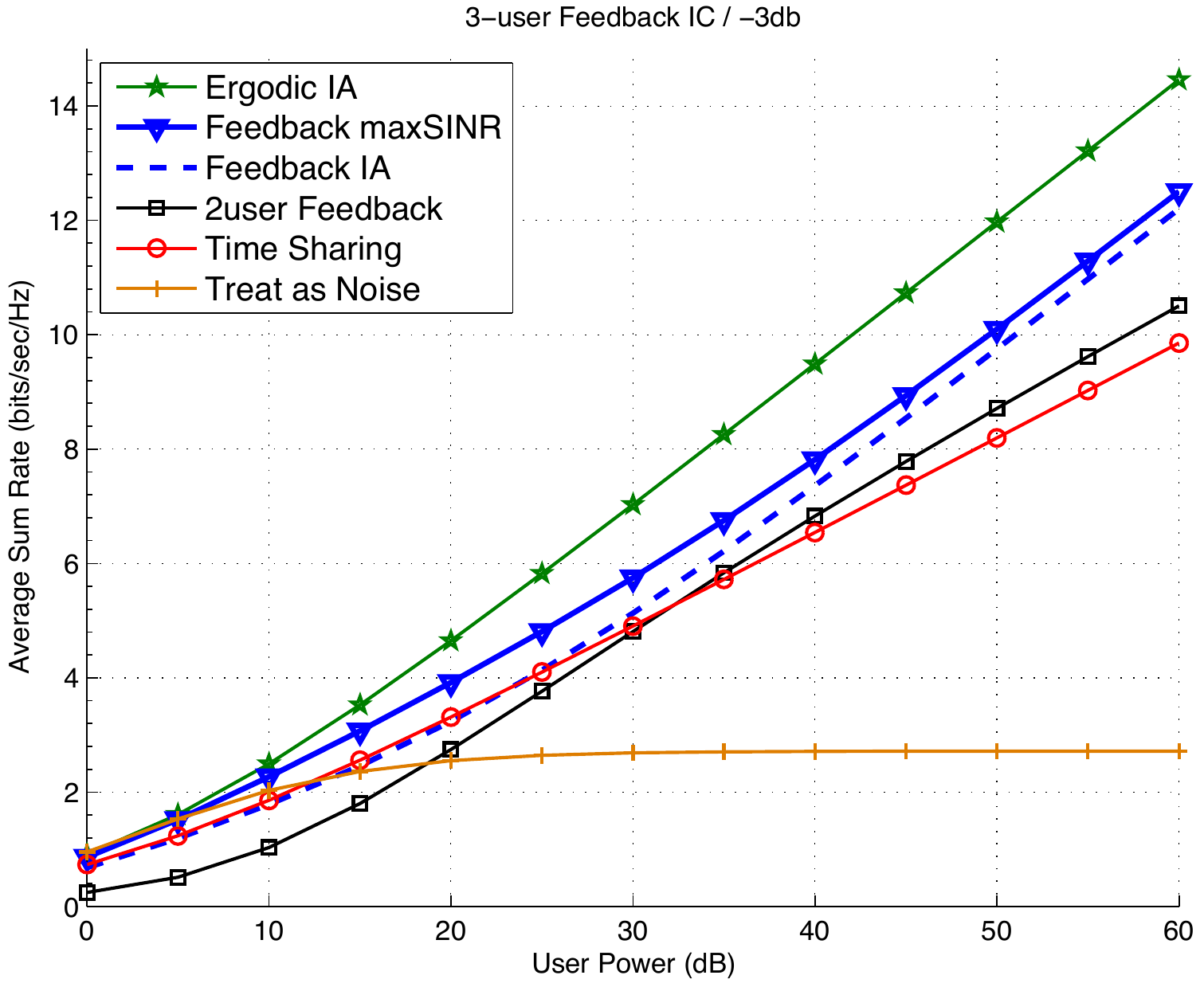}}
\caption{Performance of the proposed feedback scheme with cross-channel gains at -3dB. }
\label{fig2}
\end{figure}

\bibliographystyle{ieeetr}
\bibliography{fbalignbib}

\begin{thebibliography}{10}

\bibitem{mmk08}
M.~A. Maddah-Ali, A.~S. Motahari, and A.~K. Khandani, ``Communication over
  \textsc{MIMO} \textsc{X} channels: Interference alignment, decomposition, and
  performance analysis,'' {\em IEEE Transactions on Information Theory},
  vol.~54, pp.~3457--3470, August 2008.

\bibitem{cj08}
V.~R. Cadambe and S.~A. Jafar, ``Interference alignment and the degrees of
  freedom for the \textsc{K}-user interference channel,'' {\em IEEE
  Transactions on Information Theory}, vol.~54, pp.~3425--3441, August 2008.

\bibitem{jafar11}
S.~A. Jafar, ``Interference alignment - a new look at signal dimensions in a
  communication network,'' in {\em Foundations and Trends in Communications and
  Information Theory}, vol.~7, NOW Publishers, 2011.

\bibitem{mgmk09}
A.~S. Motahari, S.~O. Gharan, M.-A. Maddah-Ali, and A.~K. Khandani, ``Real
  interference alignment: Exploiting the potential of single antenna systems,''
  {\em IEEE Transactions on Information Theory}, Submitted November 2009.
\newblock Available online http://arxiv.org/abs/0908.2282.

\bibitem{ngjv12IT}
B.~Nazer, M.~Gastpar, S.~A. Jafar, and S.~Vishwanath, ``Ergodic interference
  alignment,'' {\em IEEE Transactions on Information Theory}, vol.~58,
  pp.~6355--6371, October 2012.

\bibitem{shannon56}
C.~E. Shannon, ``The zero-error capacity of a noisy channel,'' {\em IEEE
  Transactions on Information Theory}, vol.~2, pp.~8--19, September 1956.

\bibitem{ozarow84}
L.~H. Ozarow, ``The capacity of the white {G}aussian multiple access channel
  with feedback,'' {\em IEEE Transactions on Information Theory}, vol.~30,
  pp.~623--629, July 1984.

\bibitem{ozarowleung84}
L.~H. Ozarow and S.~K. Leung-Yan-Cheong, ``An achievable region and outer bound
  for the gaussian broadcast channel with feedback,'' {\em IEEE Transactions on
  Information Theory}, vol.~30, no.~4, pp.~667--671, 1984.

\bibitem{st11}
C.~Suh and D.~Tse, ``Feedback capacity of the \textsc{G}aussian interference
  channel to within 2 bits,'' {\em IEEE Transactions on Information Theory},
  vol.~57, pp.~2667--2685, May 2011.

\bibitem{cj09ITc}
V.~R. Cadambe and S.~A. Jafar, ``Degrees of freedom of wireless networks with
  relays, feedback, cooperation, and full duplex operation,'' {\em IEEE
  Transactions on Information Theory}, vol.~55, pp.~2334--2344, May 2009.

\bibitem{eo09}
R.~Etkin and E.~Ordentlich, ``The degrees-of-freedom of the \textsc{K}-user
  \textsc{G}aussian interference channel is discontinuous at rational channel
  coefficients,'' {\em IEEE Transactions on Information Theory}, vol.~55,
  pp.~4932--4946, November 2009.

\bibitem{mjs12}
H.~Maleki, S.~A. Jafar, and S.~Shamai, ``Retrospective interference alignment
  over interference networks,'' {\em IEEE Journal of Selected Topics in Signal
  Processing}, vol.~6, pp.~228--240, June 2012.

\bibitem{mt10}
M.~A. Maddah-Ali and D.~N.~C. Tse, ``Completely stale transmitter channel state
  information is still very useful,'' in {\em 48th Annual Allerton Conference
  on Communications, Control, and Computing}, (Monticello, IL), September 2010.

\bibitem{psd12}
D.~S. Papailioupoulos, C.~Suh, and A.~G.~D. Dimakis, ``Feedback in the $k$-user
  interference channel,'' in {\em Proceedings of the IEEE International
  Symposium on Information Theory (ISIT 2012)}, (Cambridge, MA), July 2012.

\bibitem{gv12}
Q.~Geng and P.~Viswanath, ``Interactive interference alignment,'' in {\em
  Proceedings of the IEEE Information Theory Workshop (ITW 2012)}, (Lausanne,
  Switzerland), September 2012.

\bibitem{nrt99}
A.~Nemirovski, C.~Roos, and T.~Terlaky, ``On maximization of quadratic form
  over intersection of ellipsoids with common center,'' {\em Mathematical
  Programming}, no.~86, pp.~463--473, 1999.

\end{thebibliography}

\end{document}